\documentclass[a4paper,onecolumn,11pt,accepted=2026-02-27]{quantumarticle}
\pdfoutput=1
\usepackage[utf8]{inputenc}
\usepackage[english]{babel}
\usepackage[T1]{fontenc}
\usepackage{amsmath}
\usepackage{hyperref}

\usepackage{tikz}
\usepackage{lipsum}

\usepackage[numbers,sort&compress]{natbib}

\newcommand{\ewise}[1]{\langle \langle #1 \rangle \rangle}

\newcommand{\stabpm}{\mathrm{Stab}_{\pm}}
\newcommand{\binentropy}{\mathsf{H}}
\newcommand{\minentropy}{\mathsf{H}_\infty}

\usepackage{amsmath, amsthm, amssymb, amsfonts, braket, complexity, enumerate, mathrsfs, tikz, caption, subcaption, xcolor, mathtools, braket, dirtytalk}

\usepackage{tabularx}

\usepackage[linesnumbered,ruled,vlined]{algorithm2e}
\SetKwInput{KwInput}{Input}                % Set the Input
\SetKwInput{KwOutput}{Output}
\SetKwFor{RepTimes}{repeat}{times}{end}
\SetArgSty{textup} % Removes italics, see https://tex.stackexchange.com/questions/193991/removing-italic-formatting-from-if-statement

\usepackage[nameinlink,capitalize]{cleveref}

\newcommand{\F}{\mathbb{F}}
\renewcommand{\hat}{\widehat}

\theoremstyle{plain}
\newtheorem{theorem}{Theorem}[section]
\newtheorem{corollary}[theorem]{Corollary}
\newtheorem{definition}[theorem]{Definition}

\newtheorem{lemma}[theorem]{Lemma}

\newcommand{\pauli}{\mathcal{P}}

\newcommand{\stabprodset}{\mathcal{SP}}

\newcommand{\eps}{\varepsilon}

\newcommand{\dist}{\mathrm{dist}}

\renewcommand{\Pr}{\mathop{\bf Pr\/}}

\newcommand{\tr}{\mathrm{tr}}

\newcommand{\calB}{\mathcal{B}}
\newcommand{\calC}{\mathcal{C}}

\newcommand{\calF}{\mathcal{F}}

\newcommand{\calY}{\mathcal{Y}}

\renewcommand{\C}{\mathbb{C}}
\newcommand{\Z}{\mathbb{Z}}

\newcommand{\abs}[1]{\lvert #1 \rvert}

\newcommand{\ketbra}[2]{\ket{#1}\!\!\bra{#2}}

\begin{document}

\title{Agnostic Tomography of Stabilizer Product States}

\author{Sabee Grewal}
\email{sabee@cs.utexas.edu}
\homepage{https://sabeegrewal.com/}
\affiliation{University of Texas at Austin, Austin, TX 78712, USA}
\orcid{0000-0002-8241-560X}

\author{Vishnu Iyer}
\email{vishnu.iyer@utexas.edu}
\homepage{https://vishnuiyer.org/}
\orcid{0000-0001-8072-1390}
\affiliation{University of Texas at Austin, Austin, TX 78712, USA}

\author{William Kretschmer}
\email{kretsch@cs.utexas.edu}
\homepage{https://wkretschmer.github.io/}
\affiliation{University of Texas at Austin, Austin, TX 78712, USA}
\affiliation{Simons Institute for the Theory of Computing, University of California, Berkeley, Berkeley, CA 94720, USA}
\orcid{0000-0002-7784-9817}

\author{Daniel Liang}
\email{daniel.liang@ll.mit.edu}
\homepage{https://daniel-you-liang.github.io/}
\affiliation{Portland State University, Portland, OR 97201, USA}
\affiliation{Rice University, Houston, TX 77005, USA}
\orcid{0000-0002-7418-0468}

\maketitle

\begin{abstract}
We define a quantum learning task called \textit{agnostic tomography}, where given copies of an arbitrary state $\rho$ and a class of quantum states $\mathcal{C}$, the goal is to output a succinct description of a state that approximates $\rho$ at least as well as any state in $\mathcal{C}$ (up to some small error $\varepsilon$).
This task generalizes ordinary quantum tomography of states in $\mathcal{C}$ and is more challenging because the learning algorithm must be robust to perturbations of $\rho$.

We give an efficient agnostic tomography algorithm for the class $\mathcal{C}$ of $n$-qubit stabilizer product states.
Assuming $\rho$ has fidelity at least $\tau$ with a stabilizer product state, the algorithm runs in time $n^{O(\log(2/\tau))} / \varepsilon^2$, which is $\mathsf{poly}(n/\varepsilon)$ for any constant $\tau$.
\end{abstract}

\section{Introduction}

\emph{Exponentiality} is at the heart of quantum mechanics. 
A system of merely $20$ qubits already requires over a million parameters to describe. 
This exponentiality is both a boon and a curse. 
On the one hand, it plays a key role in quantum speedups, such as Shor's algorithm for factoring integers~\cite{shor1997polynomial}. 
On the other hand, it seems we need an exponential number of amplitudes to describe systems of interest, preventing them from having tractable classical representations. 

A common approach to handling the exponential complexity of physically and computationally relevant systems is to search for an \textit{approximation} by a simple state (sometimes called an \textit{ansatz}). For example, there is an extensive body of work dedicated to finding approximations by product state (i.e., unentangled) ansatzes, especially for groundstates of local Hamiltonians; see~\cite{brandao2016product,kallaugher2024complexity} for an overview. 
Of course, product states are not the only possible ansatz: Hilbert space has many interesting classes of quantum states that are tractable---states that we can describe or classically simulate in polynomial time. Examples include matrix product states~\cite{garcia2007matrix} and stabilizer states~\cite{gottesman1998heisenberg,aaronson2004simulation}, which have both seen use as groundstate ansatzes~\cite{verstraete2008matrix,sun2024stabilizer}.

The problem of approximating an arbitrary state by a simpler one is equally natural in the setting of quantum learning, where the goal is to recover a classical description of an unknown quantum state from examples.
If we wish to learn an unknown state $\rho$ from nature, we cannot always assume that $\rho$ has an exact succinct classical description, but we might have an intuition that it is well-approximated by one of the families described above. 
Alternatively, if the state or learning algorithm is corrupted by even a small amount noise, it might be impossible to learn $\rho$ exactly, forcing us to settle for an approximation.

Nevertheless, the problem of efficiently learning an ansatz state approximation to an unknown $\rho$ has not seen much rigorous study in quantum algorithms, despite the ubiquity of ansatz-based approaches in fields like Hamiltonian complexity. Of course, there are many efficient quantum algorithms for learning highly expressive classes of quantum states, including matrix product states~\cite{cramer2010efficient}, (near)-stabilizer states~\cite{montanaro-bell-sampling,grewal2023efficient,leone2023learning,hangleiter2023bell}, free-fermion states~\cite{aaronson2023efficient}, low-degree phase states~\cite{arunachalam2022phase}, and outputs of constant-depth circuits~\cite{huang2024learning}. But generally speaking, these algorithms are extremely brittle: correctness requires that $\rho$ belong to the corresponding state class \textit{exactly}, and they may fail even if $\rho$ is affected by $1\%$ noise.

On the other hand, there is a rich study of a similar task in classical learning theory, known as \textit{agnostic learning}.
In agnostic learning, the goal is to find a classifier for a dataset that performs close to optimally from some fixed class of models.
The motivation for designing agnostic learning algorithms is much the same as for algorithms that approximate quantum states by simpler ones: real-world datasets rarely match abstract models exactly, and some fraction of the data could be corrupted adversarially by noise.

\subsection{Agnostic Tomography}
Inspired by classical agnostic learning, we define a learning task that captures the analogous problem for quantum states.

\begin{definition}[Agnostic tomography]
    Let $\calC$ be some class of pure quantum states. \emph{Agnostic tomography} with respect to $\calC$ is the following task: given copies of an arbitrary unknown mixed state $\rho$ and $0 < \eps < 1$, output a succinct classical description\footnote{``Succinct classical description'' is deliberately vague: the appropriate notion may depend on $\calC$ or the application in mind. For most purposes, a reasonable choice is a short quantum circuit that prepares $\ket{\phi}$.}
    of a state $\ket{\phi}$ such that
    \[
    \braket{\phi|\rho|\phi} \ge \sup_{\ket{\varphi} \in \calC} \braket{\varphi|\rho|\varphi} - \eps.
    \]
    \emph{Proper agnostic tomography} is the same, but where we additionally require that $\ket{\phi} \in \calC$.
\end{definition}

We suggest the name \textit{agnostic tomography} because quantum tomography is the task of recovering a full classical description of an unknown quantum state~\cite{FocusQuantumTomography2013}.
Agnostic tomography captures tomography of states in $\calC$ (to error $\eps$ in fidelity) as a special case: this corresponds to imposing the promise that the unknown state $\rho$ is itself in $\calC$.

Note that similar notions have appeared in prior work, albeit with different names. For example, O'Donnell and Wright~\cite{odonnell2017efficient} gave an exponential-time algorithm for learning the optimal rank-$k$ approximation of a mixed state $\rho$.
Grewal, Iyer, Kretschmer, and Liang~\cite[Theorem 1.3]{grewal2023improved} introduced an exponential-time agnostic tomography algorithm for stabilizer states, under the condition that the input state is pure.
They additionally gave a polynomial-time version of the algorithm under the promise that the optimal fidelity is at least $\cos^2(\pi/8)$.
Our definition is also close to the definition of quantum Hypothesis Selection given by B\u{a}descu and O'Donnell~\cite{buadescu2021improved}, except that we view the class $\calC$ as fixed in advance, rather than a variable input to an algorithm. A survey by Anshu and Arunachalam~\cite{anshu2023survey} further pointed out the analogy between agnostic learning and quantum Hypothesis Selection.

Informally, the goal of agnostic tomography is to find a simple ansatz from $\calC$ for a much more complicated quantum state.
The ansatz should attain a fidelity that is close to the optimal achievable within $\calC$.
Operationally, we envision experimentalists running agnostic tomography algorithms to learn an ansatz for their complicated physical system, and then working with the ansatz to extract additional meaningful information about the system.
Agnostic tomography might also find theoretical applications: for example,~\cite{grewal2023improved} suggested that an agnostic tomography algorithm for stabilizer states could be useful as a subroutine for finding low-rank stabilizer decompositions of magic states, which could in turn improve the runtime of certain classical algorithms for simulating quantum circuits~\cite{Bravyi2019simulationofquantum}.

The choice to make $\calC$ a class of pure states is not essential; this is only for simplicity.
One could straightforwardly generalize our definition to include mixed states by using the mixed state fidelity $F(\rho,\sigma) = \left(\tr{\sqrt{\sqrt{\rho} \sigma \sqrt{\rho}}}\right)^2$ as the measure of correlation.\footnote{Even more generally, one could consider other measures of correlation, such as $1$ minus the trace distance.
We chose fidelity because it is particularly easy to work with when $\calC$ consists of pure states.}
However, we have good reason to demand that algorithms operate on mixed input states $\rho$ instead of only pure states, because one of the key applications of agnostic tomography is learning in the presence of noise, and noise will generally affect the purity of a state.

We emphasize that agnostic tomography is much harder than ordinary tomography, even for extremely simple classes of states $\calC$. 
Intuitively, this is because tomography algorithms get to exploit additional structure from assuming that $\rho$ belongs to $\calC$, and this structure may be lost for states $\rho$ that are only slightly outside of $\calC$.
By contrast, agnostic tomography algorithms must be robust to small perturbations: if $\rho$ has $90\%$ fidelity with some $\ket{\phi} \in \calC$, the algorithm should detect this.
It is also important to note that agnostic tomography algorithms cannot simply assume that the input state $\rho$ is obtained from a state in $\calC$ by some fixed noise channel (e.g., depolarizing noise), and thus agnostic tomography is distinct from a variety of noisy learning tasks considered in prior work~\cite{OT22-learning,Gollakota2022hardnessofpac,poremba_et_al:LIPIcs.ITCS.2026.108,khesin2025average}.

As an illustrative example of the challenges in agnostic tomography, consider $\calC$ to be the class of $n$-qubit product states, meaning the states that decompose into a tensor product $\ket{\psi_1} \otimes \ket{\psi_2} \otimes \cdots \otimes \ket{\psi_n}$ of single-qubit states $\ket{\psi_i}$. There is a straightforward tomography algorithm to learn an unknown state $\ket{\psi}$ from $\calC$: perform tomography separately on each qubit of $\ket{\psi}$. However, this algorithm fails catastrophically to be a good agnostic learner for $\calC$. If we feed in the GHZ state $\rho = \left(\frac{\ket{0^n} + \ket{1^n}}{\sqrt{2}}\right)\left(\frac{\bra{0^n} + \bra{1^n}}{\sqrt{2}}\right)$, the local density matrix of each qubit looks maximally mixed, giving no information. Yet, the global state $\rho$ has fidelity $1/2$ with either $\ket{0^n}, \ket{1^n} \in \calC$. This further reveals that any agnostic learner for product states will have to go beyond looking at single-qubit reduced density matrices.\footnote{A more general example shows that even when the product state fidelity is arbitrarily close to $1$, single-qubit local density matrices are insufficient to find the best product state approximation; see~\cite[Section 2]{bakshi24learning}.}

Despite being a perfectly natural learning model, one reason agnostic tomography may not have gained much attention in the literature is that many quantum learning algorithms optimize sample complexity without worrying about computational efficiency~\cite{aaronson2018shadow,HKP20-classical-shadows,buadescu2021improved}.
However, agnostic tomography is often \textit{easy} from the perspective of sample complexity alone.
Concretely, Aaronson~\cite{aaronson2018shadow} showed that \textit{shadow tomography} on an unknown state $\rho$ lets one estimate $\braket{\varphi|\rho|\varphi}$ up to error $\eps$, for every $\ket{\varphi} \in \calC$, using a number of samples of $\rho$ that scales polynomially in $\log |\calC|$ and $\eps$. So, if $\calC$ is any class of $n$-qubit (pure) states that can be prepared with at most $\poly(n)$ gates (e.g., product states, stabilizer states, or matrix product states of polynomial bond dimension), we have $\log |\calC| \le \poly(n)$, and therefore $\calC$ can be agnostically learned with only $\poly(n, 1/\eps)$ samples. For this reason, the main question of interest surrounding agnostic tomography is whether we can design algorithms that are efficient in both sample complexity \textit{and} runtime.

\subsection{Our Result}

The main technical result of this work is, to our knowledge, the first efficient agnostic tomography algorithm for a nontrivial class of states. We say ``nontrivial'' because there are some uninteresting state classes that are very easy to learn agnostically. For example, let $\mathcal{C}$ be any orthonormal basis that can be measured efficiently, such as the computational basis. Then, by repeatedly measuring in the $\mathcal{C}$-basis and outputting the empirical mode, one can efficiently find the best approximation of $\rho$ by a state in $\mathcal{C}$.

We show how to agnostically learn the set of \textit{stabilizer product states}. These states are tensor products of the 6 single-qubit states $\{\ket{0}, \ket{1}, \ket{+}, \ket{-}, \ket{i}, \ket{-i}\}$ that are $\pm 1$ eigenstates of the single-qubit Pauli operators $X$, $Y$, $Z$. They are so-called because they are precisely the intersection between the set of stabilizer states and the set of product states. We learn this class agnostically in quasipolynomial time, resolving an open question raised by~\cite{grewal2023improved}:

\begin{theorem}
    \label{thm:main_informal}

    There is a proper agnostic tomography algorithm for the class
    \[
    \calC \coloneqq \{\ket{0}, \ket{1}, \ket{+}, \ket{-}, \ket{i}, \ket{-i}\}^{\otimes n}
    \]
    of $n$-qubit stabilizer product states. 
    Specifically, given copies of an unknown state $\rho$ and $0 < \eps < 1$, the algorithm outputs a classical description of a stabilizer product state $\ket{\phi}$ such that
    \[
    \braket{\phi|\rho|\phi} \ge \max_{\ket\varphi \in \calC} \braket{\varphi|\rho|\varphi} - \eps,
    \]
    with constant probability. 
    Assuming $\max_{\ket\varphi \in \calC} \braket{\varphi|\rho|\varphi} \ge \tau$, the algorithm runs in time
    \[
    \frac{n^{O(\log(2/\tau)}}{\eps^2}.
    \]
\end{theorem}

The parameter $\tau$ is taken as an auxiliary input to the algorithm.
If $\tau$ is at least a constant, then the runtime becomes polynomial in $n$ and $1/\eps$.
If $\tau$ is not known in advance, we can run the algorithm repeatedly and verify the fidelity of its output to perform a binary search over $\tau$.

Despite the apparent simplicity of stabilizer product states (compared to, say, general product states), we remark that the set of stabilizer product states is a useful class of states in its own right that appears in physically relevant scenarios.
As an example, a recent work by Bakshi, Liu, Moitra, and Tang~\cite[Theorem 1.1]{bakshi2024high} showed that sufficiently high-temperature Gibbs states of local Hamiltonians are expressible as probabilistic mixtures of stabilizer product states.

\subsection{Main Ideas}
To simplify presentation in the main body of this text, we will explain the algorithm and prove its correctness under the assumption that the input $\rho$ is a pure state $\ket{\psi}$. We explain how to prove its correctness on mixed state inputs in \Cref{sec:appendix}. This only requires changing a couple of lemmas that appear at the start of the technical sections, and does not materially affect the structure of the proof.

Our algorithm can be seen as a specialization of the algorithm due to Grewal, Iyer, Kretschmer, and Liang~\cite{grewal2023improved} that finds a stabilizer state whose fidelity with an input state is $\eps$-close to maximal.
\cite{grewal2023improved}'s algorithm is itself a generalization of Montanaro's algorithm~\cite{montanaro-bell-sampling} for learning stabilizer states exactly.
The basic premise of both algorithms is to identify an $n$-qubit stabilizer state $\ket{\phi}$ by learning its \textit{stabilizer group}---the abelian group of $2^n$ Pauli operators of which $\ket{\phi}$ is a $+1$ eigenstate.
Both algorithms achieve this via a primitive known as \textit{Bell difference sampling}, which performs a measurement on some copies of the unknown state and induces a distribution over Pauli operators.

In the case of Montanaro's algorithm~\cite{montanaro-bell-sampling}, Bell difference sampling on a stabilizer state $\ket{\psi}$ gives the uniform distribution over Pauli operators in the stabilizer group, modulo a global phase that can be ignored. So, to learn $\ket{\psi}$, it suffices to take $O(n)$ Bell difference samples and compute the stabilizer group generated by them. Overall, this runs in polynomial time.

\cite{grewal2023improved} show that if $\ket{\psi}$ has non-negligible fidelity with some stabilizer state $\ket{\phi}$, then a large fraction of the Pauli operators taken from Bell difference sampling will belong to the stabilizer group $S$ of $\ket{\phi}$ (again, modulo global phase).
Moreover, they prove that if $\ket{\phi}$ is the stabilizer state with the largest fidelity, the conditional distribution on $S$ will be ``close enough'' to uniform over $S$, in the sense that the conditional distribution cannot be concentrated on any proper subgroup of $S$.
Their algorithm for finding a stabilizer state $\ket{\phi}$ that approximates $\ket{\psi}$ is then the following: perform Bell difference sampling repeatedly on $\ket{\psi}$, iterate through all maximal subsets of the sampled Pauli operators that generate an abelian group $S$ of size $2^n$, and finally estimate the fidelity of $\ket{\psi}$ with every state $\ket{\phi}$ stabilized by $S$ (modulo phase). This procedure takes exponential time in general, but in some parameter regimes it is more efficient than a brute-force search over all stabilizer states---for example, if the stabilizer fidelity is at least a constant.

In this paper, we take advantage of the additional structure of stabilizer \emph{product} states. 
Whereas learning the stabilizer group of an arbitrary $n$-qubit stabilizer state required identifying $n$ independent generators in both~\cite{montanaro-bell-sampling} and~\cite{grewal2023improved}, the stabilizer group of a stabilizer product state can be identified up to phase from just a \textit{single} generator: the unique Pauli stabilizer of weight $n$. In other words, we only need to assign one of $X$, $Y$, or $Z$ to each qubit. Though we are very unlikely to sample the weight-$n$ generator in a single iteration of Bell difference sampling, we show how to identify it using many fewer than $n$ Bell difference samples.

As a concrete example, suppose that we obtain Bell difference samples $XIZII$, $IY\!ZY\!I$, and $IIIY\!Z$ from a $5$-qubit state $\ket{\psi}$. Imagine we are lucky enough that all of these Paulis belong to the stabilizer group $S$ (up to phase) of the stabilizer product state $\ket{\phi}$ that maximizes fidelity with $\ket{\psi}$.
The only full-weight Pauli operator consistent with these generators is $XY\!ZY\!Z$, because any two Pauli operators in $S$ must commute \textit{locally}---no pair of corresponding single-qubit Paulis can anticommute (e.g., given that $XIZII$ is a stabilizer, there can be no other stabilizers where the first qubit is $Y$ or $Z$, or where the third qubit is $X$ or $Y$). The corresponding stabilizer product states consistent with $XY\!ZY\!Z$ are precisely
\[
\calB = \{\ket{+}, \ket{-}\} \otimes \{\ket{i}, \ket{-i}\} \otimes \{\ket{0}, \ket{1}\} \otimes \{\ket{i}, \ket{-i}\} \otimes \{\ket{0}, \ket{1}\}.
\]
The $2^n$ states in $\calB$ form an orthonormal basis, so we can find the $\ket{\phi} \in \calB$ that maximizes fidelity by repeatedly measuring in this basis and outputting the mode.

So, how many Bell difference samples do we need to learn the stabilizer group $S$ corresponding to a stabilizer product state $\ket{\phi}$? As a first intuition, if the Bell difference samples are uniform over $S$, then the Pauli sampled at each qubit equals $I$ with $1/2$ probability, or else equals some non-identity Pauli $X$, $Y$, or $Z$ with $1/2$ probability. By the union bound, then, roughly $\log n$ samples suffice to see a non-identity Pauli on all $n$ qubits, and hence to learn $S$. This contrasts with the case of general stabilizer states, where~\cite{montanaro-bell-sampling,grewal2023improved} needed to learn $n$ linearly independent Pauli generators from $S$.

This intuition turns out to be very close to the truth: we show, by an entropy counting argument, that having $O(\log n)$ Bell difference samples from $S$ suffices (with high probability) to see non-identity Paulis on all but at most $O(\log (1/\tau))$ of the qubits, where $\tau = \abs{\braket{\phi|\psi}}^2$. The Pauli operators over the remaining qubits can be brute-forced in time $3^{O(\log (1/\tau))} = \poly(1/\tau)$.

The one remaining challenge to overcome is that not all of the Bell difference samples will belong to $S$, but only a subset. 
Because we have no way of knowing \textit{a priori} which of the Bell difference samples belong to $S$, we simply check all locally commuting subsets of size $O(\log n)$.
If $p$ is the probability that a Bell difference sample from $\ket{\psi}$ belongs to $S$, we need $O(\log n / p)$ Bell difference samples for at least $\Omega(\log n )$ of them to come from $S$ with high probability. 
By a lemma of~\cite{grewal2023improved}, we argue that $p \ge \tau^4$, which gives an overall runtime bound of roughly
\[
\binom{\log n / \tau^4}{\log n} \poly(1/\tau) \le n^{O(\log(2/\tau))}.
\]

One other noteworthy feature of the algorithm is a runtime dependence on a parameter $b$ that appears in the entropy counting argument.
The algorithm is correct for any choice of $b$ strictly between $1/2$ and $1$, and in the runtime analysis we can take $b$ to be a constant (say, $2/3$).
A closed form expression for the optimal choice of $b$ in terms of the other parameters is unlikely to exist. For this reason, we expect that in a practical implementation of the algorithm, it would be best to optimize $b$ numerically.

\subsection{Follow-up Work}

Since the initial posting of this work, the agnostic tomography model has motivated the exploration of agnostic learning algorithms for other classically-tractable quantum state families.
For example, \cite{chen2024stabilizer} gave efficient agnostic tomography algorithms for two classes that strictly generalize the set of stabilizer product states, namely: the class of arbitrary stabilizer states, and any class of ``discrete'' product states.
\cite{bakshi24learning} also independently obtained an efficient agnostic learner for discrete product states and designed several agnostic tomography algorithms for general product states.
More recently, \cite{arulandu2025agnosticproductmixedstate} gave an agnostic tomography algorithm for the class of product \textit{mixed} states by establishing a formal connection to robust statistics.
These works rely on techniques that differ substantially from ours.
Taken together, these various algorithms demonstrate the richness of the agnostic tomography model and the potential to achieve efficiency for yet more expressive classes of quantum states.

Beyond these algorithms, agnostic tomography has inspired new developments in other contexts.
\cite{wadwha2024agnostic} initiated the study of agnostic \emph{process} tomography, which aims to recover a structured quantum channel that closely approximates an unknown one.
More recent work by Bri\"et and Castro-Silva~\cite{briet2025near} used the ideas behind agnostic tomography to give new algorithms for some purely classical problems, including a near-optimal quadratic generalization of the Goldreich--Levin algorithm~\cite{levin89hardcore}.
As corollaries, they obtain a fast self-corrector for quadratic Reed--Muller codes and an algorithmic version of the Gowers Inverse Theorem.
Using related ideas, \cite{Arunachalam2025algorithmic} turned agnostic tomography procedures for stabilizer states into an algorithmic version of the Freiman--Ruzsa theorem in combinatorics.
Most recently, \cite{arunachalam2025efficiently} obtained faster quantum algorithms for PAC learning depth-$3$ $\mathsf{AC^0}$ circuits by way of a boosting method for agnostic tomography algorithms.

\section{Preliminaries}
Throughout, $\log_b$ is the logarithm base $b$, and $\log$ is shorthand for the natural logarithm $\log_e$ where $e \approx 2.718$.

\subsection{Probability}

For a discrete random variable $X$ with probability mass function $p$,
\[
\binentropy(X) \coloneqq \sum_x -p(x) \log_2(p(x))
\]
is the binary entropy of $X$. For $p \in [0,1]$, we also denote by
\[
\binentropy(b) \coloneqq -p \log_2(p) - (1-p)\log_2(1-p)
\]
the entropy of a Bernoulli random random variable with success probability $p$. The min-entropy of $X$ is denoted by $\minentropy(X) \coloneqq \log_2(1 / p_{\max})$, where $p_{\max} = \max_x p(x)$.

We use both the binary and min-entropy in the lemma below.
It says that if we take a logarithmic number of samples from a well-spread distribution over $\{0,1\}^n$, we are likely to see $1$s in all but a logarithmic number of positions.
The reader should note that the bound on $|M|$ is essentially tight: we cannot in general hope for anything bigger than $n - \log_2(C)$, because $D$ might be supported on strings that start with $\log_2(C)$ zeros.

\begin{lemma}
\label{lem:entropy_counting}
Let $D$ be a probability distribution over $\{0,1\}^n$ with $D(x) \le \frac{C}{2^n}$ for all $x$.
For some $1/2 < b < 1$, suppose we make $k \ge \log_{1/b}(n/\delta)$ independent draws from $D$ and let $M \subseteq [n]$ be the set of positions in which a $1$ is seen.
Then with probability at least $1 - \delta$, $|M| \ge n - \frac{\log_2(C)}{1 - \binentropy(b)}$.
\end{lemma}

\begin{proof}
    The entropy of $D$ satisfies
    \[
    \binentropy(D) \ge \minentropy(D) \ge -\log_2\left(\frac{C}{2^n}\right) = n - \log_2 (C).
    \]
    By the sub-additivity of entropy, we have:
    \[
    \sum_{j=1}^n \binentropy(D_j) \ge \binentropy(D) \ge n - \log_2(C),
    \]
    where $D_j$ denotes the $j$th bit of $D$.
    
    Let $J \coloneqq \{j \in [n] : \binentropy(D_j) \ge \binentropy(b)\}$.
    Since $\binentropy(D_j) \le 1$ for all $j$, we know that $|J| \ge n - \frac{\log_2(C)}{1 - \binentropy(b)}$.
    We will now prove the claim by showing that the positions corresponding to indices in $J$ are unlikely to all be $0$.
    If we denote the $k$ samples by $x_1,\ldots,x_k$, then:
    \begin{align*}
        \Pr\left[|M| \le n - \frac{\log_2(C)}{1 - \binentropy(b)}\right]
        &\le \Pr[\exists j \in J, \forall i \in [k] : x_{i,j} = 0]\\
        &\le |J| \max_{j \in J} \Pr[\forall i \in [k] : x_{i,j} = 0] && \text{(Union bound)}\\
        &\le |J|b^k && (\Pr[x_{i,j} = 0] \le b)\\
        &\le np^k\\
        &\le \delta,
    \end{align*}
    which completes the proof.
\end{proof}

Let $p$ and $q$ be distributions over $[N]$.
The \textit{$\ell_\infty$ distance} between $p$ and $q$ is defined by
\[
\dist_{\ell_\infty}(p, q) \coloneqq \max_{i \in [N]} \abs{p(i) - q(i)}.
\]

The next lemma quantifies how many samples from a categorical distribution are needed to estimate its parameters. Note the independence of the number of categories.

\begin{lemma}[Based on {\cite[Theorem 9]{Can20}}]
\label{lem:infinity_estimation}
    Let $p$ be an arbitrary distribution over $[N]$, and let $\hat{p}$ denote the empirical distribution derived from $\frac{2\log(2/\delta)}{\eps^2}$ independent samples from $p$. Then
    \[
    \Pr\left[\dist_{\ell_\infty}(\hat{p}, p) > \eps\right] \le \delta.
    \]
\end{lemma}
\begin{proof}[Proof sketch]
Letting $m = \frac{2 \log(2/\delta)}{\eps^2}$, we have
    \begin{align*}
        \Pr\left[\dist_{\ell_\infty}(\hat{p}, p) > \eps\right] &\le \Pr\left[\dist_K(\hat{p}, p) > \eps/2\right]\\
        &\le 2e^{-m \eps^2 / 2}\\
        &\le \delta,
    \end{align*}
    where $\dist_K$ is the \textit{Kolmogorov distance}, and the second line is the DKW inequality~\cite{DKW56}.
\end{proof}

\subsection{Pauli Operators}

We denote by $\pauli_n \coloneqq \{I, X, Y, Z\}^{\otimes n}$ the set of tensor products of the standard single-qubit Pauli operators $I, X, Y, Z$. The set $\overline{\pauli_n} \coloneqq \{\pm 1, \pm i\} \times \pauli_n$, the closure of $\pauli_n$ under multiplication by $\pm 1$ and $\pm i$ phases, forms a group under multiplication called the \emph{Pauli group}.

$\pauli_n$ forms an orthonormal basis for the Hermitian matrices in $\C^{2^n \times 2^n}$ under the normalized Hilbert--Schmidt inner product: $\langle A, B \rangle = 2^{-n}\tr(A^\dagger B)$. As such, for any $n$-qubit state $\ket{\psi}$ we can write
\[
\ketbra{\psi}{\psi} = \frac{1}{2^n} \sum_{P \in \pauli_n} \braket{\psi|P|\psi} P
\]
Orthonormality of $\pauli_n$ implies that $p_\psi(P) \coloneqq \frac{1}{2^n}\braket{\psi|P|\psi}^2$ forms a distribution over $\pauli_n$.
Gross, Nezami, and Walter showed that one can sample from the distribution\footnote{The multiplication $PQ$ is understood to be modulo global phase.}
\begin{equation}
\label{eq:q_psi_as_convolution}
q_{\psi}(P) \coloneqq \sum_{Q \in \pauli_n}p_\psi(Q) p_\psi(PQ)
\end{equation}
in $O(n)$ time using $4$ copies of $\ket{\psi}$, via a procedure known as \emph{Bell difference sampling}~\cite{gross2021schur}.\footnote{\cite{gross2021schur} uses slightly different notation, identifying $\pauli_n$ with a $2n$-dimensional symplectic vector space over $\F_2$. We do not use this notation here because the group structure of Pauli operators turns out to be unimportant for us, though we will introduce this notation in \Cref{sec:appendix}.} 
$q_\psi$ is a distribution because it equals the convolution $p_\psi * p_\psi$.
Bell difference sampling will be a crucial primitive in our algorithm.

\subsection{Stabilizers}

Let $\ket{\psi}$ be an $n$-qubit pure state. We define the \emph{unsigned stabilizer group} of $\ket{\psi}$ to be
\[
\stabpm(\ket{\psi}) \coloneqq \{P \in \pauli_n : P\ket{\psi} = \pm\! \ket{\psi} \}.
\]
A \emph{stabilizer state} is a state $\ket{\psi}$ satisfying $\abs{\stabpm(\ket\psi)} = 2^n$.
A \emph{stabilizer product state} is a stabilizer state that is also a tensor product of single-qubit states.
The set of $n$-qubit stabilizer product states is precisely
\[
\stabprodset \coloneqq \{\ket{0}, \ket{1}, \ket{+}, \ket{-}, \ket{i}, \ket{-i}\}^{\otimes n}.
\]
Observe that the single-qubit stabilizer states are the eigenstates of single-qubit Pauli operators:
\begin{align*}
X\ket{+} = \ket{+},
&&
X\ket{-} = -\!\ket{-},
&&
Y\ket{i} = \ket{i},
&&
Y\ket{-i} = -\!\ket{-i},
&&
Z\ket{0} = \ket{0},
&&
Z\ket{1} = -\!\ket{1}.
\end{align*}
The \emph{stabilizer product fidelity} of a state $\ket{\psi}$ is the largest fidelity of $\ket{\psi}$ with any stabilizer product state, and is denoted by
\[
\calF_\stabprodset(\ket{\psi}) \coloneqq \max_{\ket{\phi} \in \stabprodset} \abs{\braket{\phi|\psi}}^2.
\]
We say that an unsigned stabilizer group $S$ is a \emph{stabilizer product group} if $S$ has the form
\[
S = \{I, P_1\} \otimes \{I, P_2\} \otimes \cdots \otimes \{I, P_n\}
\]
for some single-qubit non-identity Pauli operators $P_1,\ldots,P_n \in \{X, Y, Z\}$.
Note that for any stabilizer product state $\ket{\psi}$, $\stabpm(\ket{\psi})$ is a stabilizer product group.
Conversely, every stabilizer product group $S$ is the unsigned stabilizer group of exactly $2^n$ different stabilizer product states $\ket{\psi}$. These states form a basis called the \emph{$S$-basis}.
The $S$-basis can be measured efficiently given generators for $S$. For example, if 
\[
S = \{I, X\} \otimes \{I, Y\} \otimes \{I, Z\},
\]
the corresponding basis of stabilizer product states stabilized by $S$ is
\[
\{\ket{+}, \ket{-}\} \otimes \{\ket{i}, \ket{-i}\} \otimes \{\ket{0}, \ket{1}\}.
\]
In other words, measuring in the $S$-basis means measuring each qubit according to the non-identity single-qubit Paulis assigned to each qubit by $S$.

\subsection{Additional Notation}
Let $P = P_1 \otimes \cdots \otimes P_n$ and $Q = Q_1 \otimes ... \otimes Q_n$ be $n$-qubit Pauli operators. We say $P$ and $Q$ \emph{locally} commute if for all $i \in [n]$, $P_iQ_i = Q_iP_i$.
Clearly, if $P$ and $Q$ both belong to a stabilizer product group, then they locally commute.

For $P,Q \in \pauli_n$, define
\[
[[P,Q]] = \abs{\{i \in [n] : P_iQ_i \neq Q_iP_i\}}.
\]
Then $P$ and $Q$ locally commute if and only if $[[P,Q]] = 0$.

Suppose a set $S \subset \pauli_n$ satisfies $[[P, Q]] = 0$ for all $P, Q \in S$.
Let $T$ be the set of indices $i \in [n]$ such that for at least one $P \in S$, $P_i \neq I$.
Then for every $i \in T$, there exists a unique $W_i \in \{X, Y, Z\}$ such that for all $x \in S$, $P_i \in \{I, W_i\}$.
For $i \not\in T$, define $W_i \coloneqq I$.
We define the \emph{local span} of $S$ to be
\[
\ewise{S} \coloneqq \{I, W_1\} \otimes \{I, W_2\} \otimes \cdots \otimes \{I, W_n\}.
\]
Note that $\ewise{S}$ is a group of size $2^{|T|}$, and every $P, Q \in \ewise{S}$ satisfies $[[P, Q]] = 0$. Additionally, there are exactly $3^{n - |T|}$ stabilizer product groups $\overline{S}$ that extend $\ewise{S}$ (i.e., for which $S \subseteq \overline{S}$): these are obtained by assigning one of $\{X, Y, Z\}$ to each index $i \not\in T$. An intuitive definition of the local span is that it is the smallest tensor product of single-qubit commuting groups that contains every Pauli operator in $S$.

\section{Algorithm and Analysis}

We are now ready to state the algorithm. In the subsections that follow, we establish its correctness and bound its runtime. As noted in the introduction, the input parameter $b$ is arbitrary, but it affects the runtime; see \Cref{sec:time_complexity} for further discussion.

\begin{algorithm}[H]
\caption{Stabilizer Product State Approximation}
\label{alg:fidelity-estimation}
\SetKwInOut{Promise}{Promise}
\KwInput{Copies of $\ket\psi$, $0 < \eps \le \tau < 1$, $1/2 < b < 1$}
\Promise{$\ket\psi$ has fidelity at least $\tau$ with some stabilizer product state}
\KwOutput{A stabilizer product state $\ket\phi$ such that $\abs{\braket{\phi|\psi}}^2 \geq \calF_\stabprodset(\ket{\psi}) - \eps$ with probability at least $1/8$}

$k \coloneqq \log_{1/b}(2n)$

$m_\mathsf{clique} \coloneqq \frac{1}{\tau^4}(k + 1)$

$t \coloneqq \frac{4\log_2(1/\tau)}{1 - \binentropy(b)}$

$m_\mathsf{est} \coloneqq \frac{8\log\left(16 \binom{m_\mathsf{clique}}{k} 3^t\right)}{\eps^2}$ 

Initialize an empty graph $G$

\RepTimes{$m_\mathsf{clique}$}{\label{line:bell_loop}
Using $4$ copies of $\ket{\psi}$, perform Bell difference sampling to obtain $W \in \pauli_n$

Add a vertex for $W$ in $G$ and an edge to all vertices $W'$ in $G$ such that $[[W,W']] = 0$
}

\ForEach{clique $(V_1, \ldots, V_k) \in G$ of size $k$}{ \label{line:clique-search}

Compute $S \coloneqq \ewise{V_1, \ldots, V_k}$

\If{$\log_2 \abs{S} \ge n - t$}{\label{line:dim_check}

\ForEach{stabilizer product group $\overline{S}$ extending $S$}{\label{line:product_extension}
Measure $m_\mathsf{est}$ copies of $\ket{\psi}$ in the $\overline{S}$-basis

\ForEach{$\ket{\phi}$ in the $\overline{S}$-basis}{\label{line:innermost}
Compute the empirical probability $\hat{o}_{\phi}$ with which $\ket{\phi}$ was measured
}

}

}
}

\Return{
whichever $\ket{\phi}$ maximizes $\hat{o}_\phi$
}

\end{algorithm}

Note that the success probability in \Cref{alg:fidelity-estimation} can be amplified from $1/8$ to $1 - \delta$ in a black-box fashion, by repeating the algorithm $O(\log(1/\delta))$ times and checking which outputs $\ket{\phi}$ have largest fidelity with $\ket{\psi}$; we omit the details.

\subsection{Correctness}
\label{sec:correctness}

We begin with two basic lemmas about the Bell difference sampling distribution. These will turn out to be the \textit{only} two lemmas that must be generalized in \Cref{sec:appendix} in order to prove \Cref{alg:fidelity-estimation}'s correctness on mixed state inputs. The first shows that Bell difference sampling cannot place too much probability on any one Pauli operator.

\begin{lemma}\label{fact:q-psi-trivial-upper}
    For all $n$-qubit quantum states $\ket \psi$ and $P \in \pauli_n$,
    \[q_\psi(P) \leq \frac{1}{2^n}.\]
\end{lemma}
\begin{proof}
    This is a simple consequence of $q_\psi$ being a convolution of $p_\psi$ with itself:
    \begin{align*}
        q_\psi(P) &= \sum_{Q \in \pauli_n} p_\psi(Q) p_\psi(PQ)\\
        &\leq \frac{1}{2^n} \sum_{Q \in \pauli_n} p_\psi(Q) && (\forall P \in \pauli_n,\,p_\psi(P) \leq \frac{1}{2^n})\\
        &= \frac{1}{2^n}. && \qedhere
    \end{align*}
\end{proof}

This next lemma shows that if a state $\ket{\psi}$ has large fidelity with a stabilizer state $\ket{\phi}$, then Bell difference samples have a good probability of belonging to the unsigned stabilizer group of $\ket{\phi}$. A conceptually similar statement was already proved in~\cite[Lemma 5.6]{grewal2023improved}.

\begin{lemma}
\label{lem:fidelity_q}
    Let $\ket{\psi}$ have fidelity at least $\tau$ with stabilizer state $\ket{\phi}$ and let $S = \stabpm(\ket{\phi})$. Then
    \[
    \sum_{P \in S} q_\psi(P) \geq \tau^4.
    \]
\end{lemma}

\begin{proof}
\begin{align*}
    \sum_{P \in S} q_\psi(P) &= \sum_{P \in S}\sum_{Q \in \pauli_n} p_\psi(Q) p_\psi(PQ) & \\
    &\geq \sum_{P \in S}\sum_{Q \in S} p_\psi(Q) p_\psi(PQ) & \\
    &= \left(\sum_{P \in S} p_\psi(P) \right)^2 & \\
    &\geq \tau^4, & (\mathrm{\cite[Lemma\ 5.2]{grewal2023improved}})\\
\end{align*}
where we have used the fact that $S$ is a subgroup of $\pauli_n$, modulo phase.
\end{proof}

We now prove the lemma that is most central to our algorithm's correctness. Below, think of $b$ as some constant (say, $2/3$). The lemma shows that if we take $O(\log n)$ Bell difference samples, then conditioned on all belonging to the unsigned stabilizer group $S$ of a stabilizer product state of fidelity $\tau$, the local span of the samples is likely to capture all but $O(\log(1/\tau))$ generators of $S$.

\begin{lemma}
\label{lem:sampling_conditional}
    Let $\ket{\psi}$ have fidelity at least $\tau$ with stabilizer product state $\ket{\phi}$. Let $S = \stabpm(\ket{\phi})$. Let $V_1, \ldots, V_k \sim q_\psi \mid S$, meaning that each $V_i$ is sampled from $q_\psi$ conditioned on being in $S$. Fix $1/2 < b < 1$. If $k \ge \log_{1/b}(n/\delta)$, then with probability at least $1 - \delta$, $\log_2 \abs{ \ewise{V_1, \ldots, V_k}} \ge n - \frac{4 \log_2(1/\tau)}{1 - \binentropy(b)}$.
\end{lemma}

\begin{proof}
    Let $\calY$ denote a draw from $q_\psi \mid S$.
    Observe that for any $P \in S$, we have
    \begin{align*}
    \Pr_{V \sim \calY}[V = P]
    & = \frac{q_\psi(P)}{\sum_{P' \in S} q_\psi(P')}\\
    &\le \frac{1}{2^n \sum_{P' \in S} q_\psi(P')} && (\mathrm{\Cref{fact:q-psi-trivial-upper}})\\
    &\le \frac{1}{\tau^4 2^n} && (\mathrm{\Cref{lem:fidelity_q}}).
    \end{align*}
    Now, let $D$ be the distribution over $\{0,1\}^n$ obtained by sampling $V \sim \mathcal{Y}$ and putting $1$s at the positions where $V$ has a non-identity Pauli.
    Then for all $x \in \{0,1\}^n$, $D(x) \le \frac{1}{\tau^4 2^n}$ because $\mathcal{Y}$ is supported over the stabilizer product group $S$, and the mapping $\mathcal{Y}$ to $D$ is a bijection between $S$ and $\{0,1\}^n$.
    By \Cref{lem:entropy_counting}, with probability at least $1 - \delta$, $k$ independent draws from $D$ will see $1$s in at least $n - \frac{4\log_2(1/\tau)}{1 - \binentropy(b)}$ locations.
    Therefore, the local span $\ewise{V_1, \ldots, V_k}$ of the corresponding samples from $\mathcal{Y}$ will generate a group of order at least $2^{n - \frac{4\log_2(1/\tau)}{1 - \binentropy(b)}}$.
\end{proof}

\Cref{lem:fidelity_q} guarantees that at least $\tau^4$ of the $q_\psi$ mass is on the stabilizer group $S^*$ of the fidelity-maximizing stabilizer product state. So, it suffices to take $m_\mathsf{clique} = O(\log n / \tau^4)$ Bell difference samples to have a good chance that some subset of size $\Omega(\log n)$ all belong to $S$. 
Hence, we can appeal to \Cref{lem:sampling_conditional} to argue that the algorithm has a decent chance of finding $S^*$ in one of its iterations.

\begin{lemma}
    \label{lem:inner_succeeds_1/4}
    Let $\ket{\varphi}$ be the stabilizer product state that has the largest fidelity with $\ket{\psi}$, and let $S^* = \stabpm(\ket{\varphi})$. Then with probability at least $\frac{1}{4}$, $\ket{\varphi}$ appears as one of the states $\ket{\phi}$ in \Cref{line:innermost}.
\end{lemma}

\begin{proof}
    Let $W_1,\ldots,W_{m_\mathsf{clique}}$ be the Bell difference samples obtained in \Cref{line:bell_loop}. Applying \Cref{lem:fidelity_q}, each $W_i$ has probability at least $\tau^4$ of belonging to $S^*$. Hence, with probability at least $1/2$, there exists a subset $\{V_1,\ldots,V_k\} \subset \{W_1,\ldots,W_{m_\mathsf{clique}}\}$ such that $V_i \in S^*$ for every $i \in [k]$, because the median of a binomial distribution with $N = m_\mathsf{clique}$ trials and rate $P \ge \tau^4$ is at least $\lfloor NP \rfloor \ge k$.
    
    \Cref{lem:sampling_conditional} additionally tells us that conditioned on $V_i \in S^*$ for every $i$, we have $\log_2 \abs{\ewise{V_1,\ldots,V_k}} \ge n - t$ with probability at least $1/2$. So, with overall probability at least $1/4$, there exists a subset $\{V_1,\ldots,V_k\} \subset \{W_1,\ldots,W_{m_\mathsf{clique}}\}$ such that $V_i \in S^*$ for every $i$ \textit{and} $\log_2 \abs{\ewise{V_1,\ldots,V_k}} \ge n - t$. 
    
    Assume henceforth that this has occurred. Because $V_i \in S^*$ for every $i$, and because $S^*$ is a stabilizer product group, $[[V_i,V_j]] = 0$ for each $i, j$. Therefore, $(V_i,\ldots,V_k)$ will be one of the cliques in the graph found in \Cref{line:clique-search}, and the subspace $S = \ewise{V_1,\ldots,V_k}$ will pass the check at \Cref{line:dim_check}. For some choice of $\overline{S}$ in \Cref{line:product_extension} we will have $\overline{S} = S^*$, so one of the states $\ket{\phi}$ in the $\overline{S}$-basis will equal $\ket{\varphi}$.
\end{proof}

We can now complete the proof of \Cref{alg:fidelity-estimation}'s correctness.

\begin{theorem}
\label{thm:alg-correctness}
    \Cref{alg:fidelity-estimation} finds a stabilizer product state $\ket\phi$ such that $\abs{\braket{\phi|\psi}}^2 \geq \calF_\stabprodset(\ket{\psi}) - \eps$ with probability at least $1/8$.
\end{theorem}

\begin{proof}
\Cref{lem:infinity_estimation} implies that in any fixed iteration of \Cref{line:product_extension}, all of the estimates $\hat{o}_\phi$ satisfy
\[
\abs{\hat{o}_\phi - \abs{\braket{\phi|\psi}}^2} \le \eps / 2,
\]
except with failure probability at most
\[
\frac{1}{8 \binom{m_\mathsf{clique}}{k} 3^t}.
\]
The loop in \Cref{line:product_extension} is executed at most $\binom{m_\mathsf{clique}}{k} 3^t$ times, because:
\begin{itemize}
    \item \Cref{line:clique-search} runs at most $\binom{m_\mathsf{clique}}{k}$ iterations;
    \item \Cref{line:product_extension} runs at most $3^t$ iterations per iteration of \Cref{line:clique-search}, because there are at most $3^t$ stabilizer product groups $\overline{S}$ that extend $S$ with size $2^{n - t}$ (this corresponds to assigning one of Pauli $X$, $Y$, or $Z$ to each unassigned qubit).
\end{itemize}
Hence, by the union bound, \textit{all} of the estimates $\hat{o}_\phi$ satisfy
\begin{equation}
\label{eq:estimates_correct}
\abs{\hat{o}_\phi - \abs{\braket{\phi|\psi}}^2} \le \eps / 2,
\end{equation}
except with failure probability at most $1/8$.

Let $\ket{\varphi}$ be the stabilizer product state that has the largest fidelity with $\ket{\psi}$, meaning $\abs{\braket{\varphi|\psi}}^2 = \calF_\stabprodset(\ket{\psi})$. By \Cref{lem:inner_succeeds_1/4}, $\ket{\varphi}$ appears as at least one of the states $\ket{\phi}$ in \Cref{line:innermost}, except with failure probability at most $3/4$.

Assume that both of the above steps succeed, which occurs with probability at least $1/8$ by the union bound. Letting $\ket{\phi}$ be the state that maximizes $\hat{o}_\phi$, we have
\begin{align*}
    \abs{\braket{\phi|\psi}}^2 &\ge \hat{o}_\phi - \eps/2 && (\mathrm{\Cref{eq:estimates_correct}})\\
    &\ge \hat{o}_\varphi - \eps/2 && (\ket{\phi}\text{ maximizes } \hat{o}_\phi)\\
    &\ge \abs{\braket{\varphi|\psi}}^2 - \eps && (\mathrm{\Cref{eq:estimates_correct}})\\
    &= \calF_\stabprodset(\ket{\psi}) - \eps,
\end{align*}
which shows that the algorithm is correct.
\end{proof}

\subsection{Time Complexity}
\label{sec:time_complexity}

The runtime of \Cref{alg:fidelity-estimation} is easily seen to be dominated by computing the estimates $\hat{o}_\phi$ in \Cref{line:innermost}. Of course, \Cref{line:innermost} should not literally be implemented by iterating through all $2^n$ states $\ket{\phi}$ in the $\overline{S}$-basis, because in general $m_\mathsf{est} \ll 2^n$---it is more efficient to maintain a dictionary of the states $\ket{\phi}$ that were measured, along with their measurement counts. If implemented in this fashion, then the overall runtime of the algorithm is bounded by
\[
\binom{m_\mathsf{clique}}{k} 3^t m_\mathsf{est} \cdot O(n)
\]
because \Cref{line:clique-search} runs at most $\binom{m_\mathsf{clique}}{k}$ iterations, \Cref{line:product_extension} runs at most $3^t$ iterations (as argued in the proof of \Cref{thm:alg-correctness} above), \Cref{line:innermost} runs $m_\mathsf{est}$ iterations, and each iteration of \Cref{line:innermost} takes $O(n)$ time to perform the measurement in a stabilizer product basis. Plugging in $m_\mathsf{est}$, the runtime becomes
\begin{align}
    \label{eq:runtime_1}
    \binom{m_\mathsf{clique}}{k} 3^t \log \left(\binom{m_\mathsf{clique}}{k} 3^t\right) \cdot O(n/\eps^2).
\end{align}
If we take $b$ to be some constant such as $2/3$, then $3^t \le \poly(1/\tau)$ and
\[
\binom{m_\mathsf{clique}}{k} \le \binom{O(\log n / \tau^4)}{O(\log n)} \le \poly(1/\tau)^{O(\log n)} \le n^{O(\log(2/\tau))}.
\]
Therefore, the total runtime of the algorithm is at most
\begin{align*}
    \frac{n^{O(\log(2/\tau))}}{\eps^2}.
\end{align*}
Note that the factor $2$ in the $O(\log(2/\tau))$ is an arbitrary constant greater than $1$, and is only necessary when $\tau$ is sufficiently close to $1$.
If we assume $\tau$ at most some constant bounded away from $1$, it becomes correct to write $O(\log(1/\tau)$ in the exponent.

A bound on the explicit constant in the exponent requires more careful accounting.
Consider
\begin{align*}
    f(n,\tau,b) &\coloneqq \binom{m_\mathsf{clique}}{k} 3^t,
\end{align*}
so that the runtime becomes quasilinear in $f$, i.e.,
\[
f(n,\tau,b) \log(f(n,\tau,b)) \cdot O(n/\eps^2).
\]
First we bound the term corresponding to \Cref{line:clique-search}:
\begin{align*}
\binom{m_\mathsf{clique}}{k}
&= \binom{\frac{1}{\tau^4}(\log_{1/b}(2n) + 1)}{\log_{1/b}(2n)}\\
&\le \left(\left(1 + \frac{1}{\log_{1/b}(2n)} \right)^{\log_{1/b}(2n)} \binom{\frac{1}{\tau^4}\log_{1/b}(2n)}{\log_{1/b}(2n)}\right) && \binom{a+b}{k} \le \left(\frac{a+b}{a}\right)^k\binom{a}{k}\\
&\le e \binom{\frac{1}{\tau^4}\log_{1/b}(2n)}{\log_{1/b}(2n)} && (1 + 1/x)^x \le e\\
&\le e(e\tau^{-4})^{\log_{1/b}(2n)} && \binom{a}{k} \le \left(\frac{ae}{k}\right)^k\\
&= e(2n)^{\log_{1/b}(e / \tau^4)}\\
&\le e(2n)^{4\log_{1/b}(e/\tau)}.
\end{align*}
Then
\[
f(n, \tau, b) \le e(2n)^{4\log_{1/b}(e/\tau)} 3^{\frac{4\log_2(1/\tau)}{1 - \binentropy(b)}}.
\]
Optimizing the choice of $b$ for general values of $n$ and $\tau$ seems difficult, but the bound does not suggest any clear benefit to choosing $b$ particularly close to $1/2$ or $1$ instead of some fixed constant in between.

\section{Discussion and Open Problems}
Can the runtime of \Cref{alg:fidelity-estimation} be improved to polynomial in $n$ and $1/\tau$, rather than quasipolynomial? At minimum, it might be possible to further optimize the leading constant in the exponent of the runtime.
We also suspect that the algorithm could be made more efficient in practice.
For example, rather than fixing $m_\mathsf{clique}$ in advance, one could take Bell difference samples in an online fashion until the graph $G$ contains a sufficiently large clique.
This probably occurs faster in practice than our worst-case bounds predict.

Are there other classes of quantum states with simple classical descriptions that admit efficient agnostic tomography algorithms?
Other examples to consider might include outputs of low-depth circuits and free-fermionic states.

\section*{Acknowledgments}
SG is supported (via Scott Aaronson) by a Vannevar Bush Fellowship from the US Department
of Defense, the Berkeley NSF-QLCI CIQC Center, a Simons Investigator Award, and the Simons “It
from Qubit” collaboration. 
VI is supported by an NSF Graduate Research Fellowship.
WK acknowledges support from the U.S.\ Department of Energy, Office of Science, National
Quantum Information Science Research Centers, Quantum Systems Accelerator. 
DL is supported by NSF award FET-2243659.

This work was done in part while SG, VI, and DL were visiting the Simons Institute for the Theory of Computing, supported by NSF QLCI Grant No.\ 2016245.

\bibliographystyle{quantum}
\bibliography{refs}

\appendix

\section{Agnostic Tomography of Mixed States}
\label{sec:appendix}
In this appendix, we show that \Cref{alg:fidelity-estimation} remains correct when run on mixed state inputs. To do so, we generalize some properties of Bell difference sampling (specifically, \Cref{fact:q-psi-trivial-upper,lem:fidelity_q}) to mixed states. This will require different notation for working with Pauli operators, which we now introduce.

\subsection{Weyl Operators}
We begin by identifying the unsigned Pauli operators $\pauli_n$ with a vector space $\F_2^{2n}$ via the formalism of Weyl operators, following Gross, Nezami, and Walter~\cite{gross2021schur}. For $x = (a, b) \in \F_2^{2n}$, the \emph{Weyl operator} $W_x$ is 
\[
W_x \coloneqq 
i^{a'\cdot b'}(X^{a_1} Z^{b_1}) \otimes \dots \otimes (X^{a_n} Z^{b_n}),
\]
where $a',b' \in \Z^n$ are the embeddings of $a,b$ into $\Z^n$. Each $W_x$ is an element of $\pauli_n$ and vice-versa, so this yields a one-to-one correspondence between $\F_2^{2n}$ and $\pauli_n$.

For $x, y \in \F_2^{2n}$, we call $[x,y]$ the \textit{symplectic product} of $x$ and $y$ over $\F_2^{2n}$, which is defined by
\[
[x,y] = x_1 \cdot y_{n+1} + x_2\cdot y_{n+2} + ... + x_n \cdot y_{2n} + x_{n+1} \cdot y_1 + x_{n+2} \cdot y_2 + ... + x_{2n} \cdot y_n,
\]
Operationally, the symplectic product captures commutation relations: $W_x$ and $W_y$ commute if and only if $[x, y] = 0$.

\subsection{Bell Difference Sampling Mixed States}

Performing Bell difference sampling on a mixed state $\rho$ induces a distribution $q_\rho$ over $\pauli_n$ (or equivalently, $\F_2^{2n}$), but it is a bit harder to define than for pure states.
\cite[Equation 3.7]{gross2021schur} shows that Bell difference sampling can be written as a projective measurement $\{\Pi_x\}_{x \in \F_2^{2n}}$, where
    \[\Pi_x \coloneqq \frac{1}{4^n}\sum_{a \in \F_2^{2n}} (-1)^{[x, a]}W_a^{\otimes 4}.
\]
Thus, we can write $q_\rho(x) \coloneqq \tr(\Pi_x \rho^{\otimes 4})$ as the probability that Bell difference sampling obtains $a \in \F_2^{2n}$.
We can also define $p_\rho(x) \coloneqq \tr(W_x \rho)^2/2^n$ in analogy with $p_\psi$. We note that for mixed states, unlike for pure states, $p_\rho$ is no longer a distribution (see \cref{fact:sum-of-p-mixed}), and $q_\rho$ can no longer be expressed as the convolution of $p_\rho$ with itself, as we did in \Cref{eq:q_psi_as_convolution}.
Nonetheless, we shall see that $p_\rho$ and $q_\rho$ are related in useful ways.

\subsection{Properties of Bell Difference Sampling}

We first show a connection between the purity of a state $\rho$ and $p_\rho$

\begin{lemma}\label{fact:sum-of-p-mixed}
Let $\rho$ be a quantum mixed state. Then
    \[\tr(\rho^2) = \sum_{x \in \F_2^{2n}} p_\rho(x).\]
\end{lemma}
\begin{proof}
Because the Weyl operators form an orthogonal basis for the space of Hermitian $n$-qubit operators, we can express $\rho$ in this basis as:
\begin{align*}
\rho = \frac{1}{2^{n}} \sum_{x \in \F_2^{2n}} \tr(W_x \rho)  W_x. 
\end{align*}
Consequently,
\begin{align*}
   \tr(\rho^2) 
   &=\frac{1}{4^{n}} \tr\left(\sum_{x,y \in \F_2^{2n}} \tr(W_x \rho)\tr(W_y \rho)  W_x W_y\right) \\ 
   &=\frac{1}{4^{n}} \sum_{x,y \in \F_2^{2n}} \tr(W_x \rho)\tr(W_y \rho)  \tr\left(W_x W_y\right) \\ 
   &=\frac{1}{2^n}\sum_{x \in \F_2^{2n}} \tr(W_x \rho)^2 && (\tr(W_xW_y) = 2^n \delta_{xy}) \\ 
   &= \sum_{x \in \F_2^{2n}} p_\rho(x). && \qedhere
\end{align*}
\end{proof}

With this in hand, we can establish an analogue of \Cref{fact:q-psi-trivial-upper} for mixed states.

\begin{lemma}
\label{fact:q-pho-trivial-upper-mixed}
    For all $n$-qubit mixed states $\rho$ and $x \in \F_2^{2n}$,
    \[
    q_\rho(x) \leq \frac{\tr(\rho^2)}{2^n}\leq \frac{1}{2^n}.
    \]
\end{lemma}
\begin{proof}
\begin{align*}
        q_\rho(x) 
        &= \sum_{a \in \F_2^{2n}} (-1)^{[x, a]} p_\rho(a)^2 \\ 
        &\leq \sum_{a \in \F_2^{2n}} p_\rho(a)^2 \\ 
        &\leq \frac{1}{2^n}\sum_{a \in \F_2^{2n}} p_\rho(a) \\ 
        &\leq \frac{\tr(\rho^2)}{2^n}.
\end{align*}
    Above, the first line follows from the definitions of $q_\rho$ and $p_\rho$ and the linearity of the trace. 
    The third line follows from the observation that $p_\rho(x) \leq \frac{1}{2^n}$ always. The last line is an application of \cref{fact:sum-of-p-mixed}.
\end{proof}

Finally, we prove the analogue of \Cref{lem:fidelity_q}, showing that $q_\rho$ places large mass on the unsigned stabilizer groups of large-fidelity stabilizer states.

\begin{lemma}
\label{lem:fidelity_q_mixed}
    Let $\rho$ be an $n$-qubit mixed state and let $\ket{\phi}$ be a stabilizer state
    that has fidelity $\tau$ with $\rho$, meaning $\braket{\phi|\rho|\phi} = \tau$. Let $S = \{x \in \F_2^{2n} : W_x \in \stabpm(\ket{\phi})\}$. Then
    \[
    \sum_{x \in S} q_\rho(x) \geq \tau^4.
    \]
\end{lemma}
\begin{proof}

    By definition of $q_\rho$, we have
    \begin{align*}
        \sum_{x \in S} q_\rho(x) &=
        \frac{1}{4^n} \sum_{x \in S} \sum_{a \in \F_2^{2n}} (-1)^{[x,a]} \tr\left(W_a^{\otimes 4} \rho^{\otimes 4}\right)
        \\
        &= \sum_{x \in S} \sum_{a \in \F_2^{2n}} (-1)^{[x,a]} p_\rho(a)^2
        \\
        &=  \sum_{a \in \F_2^{2n}} p_\rho(a)^2 \sum_{x \in S} (-1)^{[x,a]} 
        \end{align*}
        Now, if $a \in S$ then $[x,a] = 0$ for all $x \in S$, because all elements of a stabilizer group must commute. Otherwise, $[x,a] = 0$ for half of the $x \in S$ and $[x,a] = 1$ for the other half. This means that
        \[
        \sum_{x \in S} (-1)^{[x,a]} = 
        \begin{cases}
            2^n & a \in S,\\
            0 & a \not\in S.
        \end{cases}
        \]
        As such, we can simplify:
        \begin{align*}
        \sum_{x \in S} q_\rho(x) &= 2^n \sum_{x \in S} p_\rho(x)^2 \\
        &= \frac{1}{2^n} \sum_{x \in S} \tr\left(W_x \rho\right)^4\\
        &\ge \frac{1}{{16}^n} \left(\sum_{x \in S} \abs{\tr\left(W_x \rho\right)} \right)^4 && (\text{Jensen's Inequality})
    \end{align*}
    To complete the proof, let $\overline{S} \coloneqq \{P \in \overline{\pauli_n} : P\ket{\phi} = \ket{\phi}\}$ be the (signed) stabilizer group of $\ket{\phi}$. By well-known calculations (e.g.,~\cite[Equation (2)]{audenaert2005entanglement}), the density matrix of a stabilizer state can be written as an average over the stabilizer group:
    \begin{equation}
    \label{eq:stabilizer_sum}
    \ketbra{\phi}{\phi} = \frac{1}{2^n} \sum_{P \in \overline{S}} P.
    \end{equation}
    Then:
    \begin{align*}
        \sum_{x \in S} q_\rho(x)
        &\ge \frac{1}{{16}^n} \left(\sum_{x \in S} \abs{\tr\left(W_x \rho\right)} \right)^4\\
        &\ge \frac{1}{{16}^n} \left(\sum_{P \in \overline{S}} \tr\left(P \rho\right) \right)^4 && (\text{Triangle inequality})\\
        &= \frac{1}{16^n} \left(2^n \tr\left(\ketbra{\phi}{\phi} \rho\right)\right)^4 && (\text{Linearity of Trace and \Cref{eq:stabilizer_sum}})\\
        &= \braket{\phi|\rho|\phi}^4\\
        &= \tau^4. && \qedhere
    \end{align*}
\end{proof}

It now follows that \Cref{alg:fidelity-estimation} remains correct when run on mixed state inputs.

\begin{corollary}
    On input a mixed state $\rho$ in place of $\ket{\psi}$, \Cref{alg:fidelity-estimation} outputs a stabilizer product state $\ket\phi$ such that $\braket{\phi|\rho|\phi} \geq \max_{\ket{\varphi} \in \stabprodset} \braket{\varphi|\rho|\varphi}  - \eps$ with probability at least $1/8$.
\end{corollary}

\begin{proof}[Proof sketch]
    Follow the proof steps in \Cref{sec:correctness}, replacing \Cref{fact:q-psi-trivial-upper} with \Cref{fact:q-pho-trivial-upper-mixed}, replacing \Cref{lem:fidelity_q} with \Cref{lem:fidelity_q_mixed}, and identifying each $x \in \F_2^{2n}$ with the corresponding Weyl operator $W_x \in \pauli_n$.
\end{proof}

\end{document}